\documentclass{article}
\usepackage{spconf,amsmath,graphicx,hyperref}
\usepackage{amsmath,amsfonts}
\usepackage{algorithmic}
\usepackage{algorithm}
\usepackage{array}
\usepackage[caption=false,font=normalsize,labelfont=sf,textfont=sf]{subfig}
\usepackage{textcomp}
\usepackage{stfloats}
\usepackage{url}
\usepackage{verbatim}
\usepackage{graphicx}
\usepackage{cite}
\usepackage{comment}

\usepackage{amsmath,graphicx,hyperref}
\usepackage{amssymb}  
\usepackage{amsthm}
\usepackage{bm}      
\usepackage{booktabs}
\usepackage[normalem]{ulem}

\usepackage{xcolor}

\newtheorem{theorem}{Theorem}
      
\newtheorem{proposition}{Proposition}

\usepackage[normalem]{ulem}

\title{Power-MSE trade-off of Factorized Low-rank Approximated Computation Scheme with Memristors}

\name{
    Binyu Lu$^{\hspace{1pt} \star \hspace{1pt} \dagger}$ 
    \qquad Matthias Frey$^{\hspace{1pt} \star}$ 
    \qquad Stark Draper$^{\hspace{1pt} \dagger}$  
    \qquad Jingge Zhu$^{\hspace{1pt} \star}$ 
}

\address{
    $^{\star}$ Department of Electrical and Electronic Engineering, University of Melbourne \\
    $^{\dagger}$ Department of Electrical and Computer Engineering, University of Toronto
}

\begin{document}
\ninept
\maketitle

\begin{abstract}
\vspace{-5pt}
Memristor crossbars enable analog vector–matrix multiplication (VMM) which is promising for machine learning applications.
Scaling matrix entries to lower memristor conductance levels reduces power consumption but increases the impact of memristor programming noise on VMM accuracy.
To investigate how low-rank factorization can improve this trade--off, we extend the previously proposed factorized low-rank approximation scheme (FLAS) \cite{lu2026low} to support adjustable conductance scaling.
We then derive closed-form MSE and power expressions for both FLAS and baseline  VMM. 
Based on these expressions, we establish an analytical power--MSE trade--off framework to capture the coupled effects of approximation rank, replication allocation, and conductance scaling under constraints on memristor count and conductance scaling bounds. 
Numerical results demonstrate FLAS’s power--MSE advantage across matrices with different singular value spectra.
The power decomposition explains how TIA feedback resistance affect this advantage.
\end{abstract}

\vspace{-5pt}

\vspace{-5pt}

\vspace{-5pt}
\section{Introduction}
\vspace{-5pt}
In-memory computing (IMC) has emerged as a promising paradigm by performing computations directly where the data are stored \cite{silvano2025survey,ielmini2018memory}.
Memristor crossbars are developed as an attractive physical substrate for IMC. 
A memristor crossbar, as illustrated in Fig.~\ref{memristor}, is a two-dimensional array of memristors, where each memristor is a element with a programmable conductance \cite{yao2020fully}. 
It is able to implement vector-matrix multiplication (VMM) efficiently.
Matrix entries are represented as memristor programmed conductances.
Input vector entries are applied as row voltages. 
By Ohm's law and the Kirchhoff's current law \cite{sebastian2020memory}, the resulting column currents perform parallel multiply-and-accumulate operations, and thus realize the VMM. 
At the end of each column, a transimpedance amplifier (TIA) converts the total column current into an output voltage.
\begin{figure}[htbp]
	\centering
    \vspace{-10pt}
	\includegraphics[width = 0.8 \linewidth]{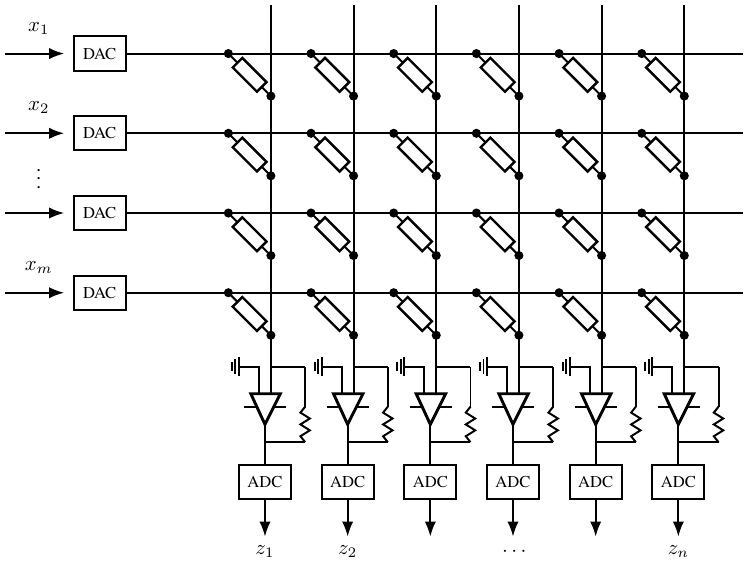}
    \vspace{-15pt}
	\caption{Memristor crossbar architecture.\hspace{30pt}}
	\label{memristor}
    \vspace{-18pt}
\end{figure}

Prior work \cite{dupraz2020noisy,dupraz2021power,kern2024fast} has developed analytical models of error propagation and power–accuracy trade-offs for directly implemented baseline VMM scheme with memristors. 
Our previous work \cite{lu2026low} proposed the factorized low-rank approximation scheme (FLAS).
FLAS improves VMM accuracy through low-rank factorization and stage-wise averaging, but did not consider conductance scaling or power consumption.
Consequently, how to introduce conductance scaling in the FLAS to make the power control possible, and whether FLAS consumes less power than baseline, are unexplored.

In this paper, we extend the FLAS for a memristor model with adjustable column-wise conductance scaling and derive closed-form mean squared error (MSE) and power consumption expressions for both FLAS and baseline VMM.
We then jointly optimize approximation rank, replication allocation, and conductance scaling to minimize power subject to a prescribed MSE tolerance, a budget on the number of memristors, and conductance scaling bounds. 
By varying the MSE tolerance, we identify the power–MSE Pareto frontier.
Numerical results show that FLAS achieves lower power for matrices with rapidly decaying singular values and retains this advantage for moderately decaying spectra when a larger output MSE is permitted. 
A comparison of memristor and TIA power contributions further explains how increasing TIA feedback resistance strengthens FLAS’s relative power advantage when larger output MSEs are permitted, but favors baseline VMM when more accurate outputs are required.

The remaining part of the paper is organized as follows.
Section~\ref{SystemModel} presents the noisy memristor crossbar architecture with column-wise conductance scaling.
Section~\ref{MSE_power_baseline} develops the MSE and power consumption analysis for baseline VMM.
Section~\ref{MSE_power_FLAS} extends FLAS to flexible column-wise conductance scaling and develops the MSE and power analysis.
Section~\ref{tradeoff} analyzes the power--MSE tradeoff and illustrate the tradeoff through  numerical examples.

\vspace{-10pt}
\section{System Model}\label{SystemModel}
\vspace{-5pt}
Memristor crossbars allow for the implementation of
the VMM 
$\mathbf c=\mathbf b \mathbf A$, where
$\mathbf b\in\mathbb R^{1\times m}$,
$\mathbf A\in\mathbb R^{m\times n}$, and
$\mathbf c\in\mathbb R^{1\times n}$.
As illustrated in Fig.~\ref{memristor}, we denote the voltage applied to the row word line $j$ as $x_j$,
and the conductance of the memristor at the intersection of row $j$ and column $k$ of an $m\times n$ memristor crossbar as $g_{j,k}$.
Then the $k$-th column bit line voltage output is
$z_k=r_{\mathrm T}\sum_{j=1}^m x_j g_{j,k}$, where $r_{\mathrm T}$ is the feedback resistance of a TIA.

We first introduce how baseline VMM scheme is implemented.
For simplicity, in this paper we assume $|b_j|_{\max}\leq V_{\mathrm{in}}^{\max}$ \cite{liu2022memristor}, where $V_{\mathrm{in}}^{\max}$ is the allowable voltage amplitude.
Consequently, we set $x_j=b_j$ for $j \in \{1,...,m \}$.
Denote the conductance range in which the memristor operates reliably as $[g_{\min},g_{\max}]$.
To fit the entries of $\mathbf A$ into this range, we apply column-wise
scaling. Let
\begin{align}
a_{j,k}^{\mathrm{sc}} \triangleq \kappa_{\mathrm A,k}a_{j,k}, \label{eq:scaled_weight}
\end{align}
denote the scaled version of $a_{j,k}$, 
where the scaling factor of column $k$ is defined as $\kappa_{\mathrm A, k} \triangleq \frac{g_{u,k}-g_{\min}}{|a_k^{\max}|}$, where $|a_k^{\max}|\triangleq\max_{1\leq j\leq m}|a_{j,k}|$, and $g_{u,k}$ denotes an adjustable upper target conductance satisfying $g_{\min}<g_{u,k}\leq g_{\max}$. We assume that $|a_k^{\max}|$ are nonzero for all $k \in \{1,...,n\}$.
The scaling matrix for $\mathbf{A}$ is defined as
\begin{equation}\label{D_A_def}
\mathbf D_{\mathrm A} \triangleq
    \operatorname{diag}
    (\kappa_{\mathrm A,1},\ldots,\kappa_{\mathrm A,n}).
\end{equation}

Since memristor conductances are nonnegative, each signed
matrix entry is represented using a positive and a negative
crossbar.
Let $s\in\{+,-\}$ index the positive and negative crossbars.
The target conductances are written as
$
g_{j,k}^{(s)}=g_{\min}+[a_{j,k}^{\mathrm{sc}}]_s, s\in\{+,-\},
$
where $[u]_+\triangleq\max(u,0)$ and
$[u]_-\triangleq\max(-u,0)$.
The positive and negative crossbars are driven by the same input voltages, and their bit-line currents are subtracted by a differential readout.
Thus, the common offset $g_{\min}$ cancels under differential
readout. 
The bit line output voltage is measured as $z_k = r_{\mathrm{T}} \sum_{j=1}^{m} x_j g_{j,k}^{\text{eff}}$, where the effective conductance is defined as $g_{j,k}^{\mathrm{eff}} \triangleq
    g_{j,k}^{(+)}-g_{j,k}^{(-)}
     =a_{j,k}^{\mathrm{sc}}$.
Then after rescaling $\hat{c}_k \triangleq \frac{z_k}{r_{\mathrm{T}}\kappa_k}$, we can recover the VMM result $\hat{c}_k $.

Due to device and circuit nonidealities, the programmed
conductances may differ from their target values. 
We model the programmed conductance value on each crossbar as a random variale $G_{j,k}^{(s)}$ based on the additive programming noise model
\cite{le2022precision,liu2018memristor,joshi2020accurate}
\begin{equation}
G_{j,k}^{(s)} =g_{j,k}^{(s)}+N_{j,k}^{(s)},
\quad s\in\{+,-\},
\label{eq:noisy_conductance_general}
\end{equation}
where $N_{j,k}^{(s)}$ denotes the random variable representing the programming noise with zero mean and variance $\tau_v^2$.
In this paper, the noise variables are assumed mutually independent across all memristors and crossbars. 
The entries $N_{j,k}^{\mathrm{eff}}$ form the effective conductance noise matrix $\mathbf N_{\mathrm A}^{\mathrm{eff}}$.
\begin{equation}
N_{j,k}^{\mathrm{eff}}\triangleq
    N_{j,k}^{(+)}-N_{j,k}^{(-)}.
\label{eq:effective_noise}
\end{equation}
Then we have $E[N_{j,k}^{\mathrm{eff}}]=0$, and $\operatorname{Var}(N_{j,k}^{\mathrm{eff}}) =2\tau_v^2$.
Therefore, for the directly implemented baseline VMM with programming noise, the analog output of array $s$ is given by
\begin{equation}
Z_k^{\mathrm{base}(s)} = r_{\mathrm T}\sum_{j=1}^m x_j G_{j,k}^{(s)},
\quad s\in\{+,-\}.
\label{eq:baseline_array_output}
\end{equation}
We define the efficient noise matrix after subtraction of positive and negative array and rescaling as
$\mathbf E_{\mathrm A}\triangleq
\mathbf N_{\mathrm A}^{\mathrm{eff}}
\mathbf D_{\mathrm A}^{-1}$, whose $(j,k)$-th entry is
$E_{j,k}=N_{j,k}^{\mathrm{eff}}/\kappa_{\mathrm A,k}$.
According to the mean and variance of $N_{j,k}^{\mathrm{eff}}$ in \eqref{eq:effective_noise}, we know $\mathbb E[E_{j,k}]=0$, and $\operatorname{Var}(E_{j,k}) = \frac{2\tau_v^2}{\kappa_{\mathrm A,k}^2}$.
The recovered output is
\begin{align}
\hat c_k^{\mathrm{base}}
    &=
    \frac{
    Z_k^{\mathrm{base}(+)}
    -Z_k^{\mathrm{base}(-)}
    }{r_{\mathrm T}\kappa_{\mathrm A,k}}
    =
    \sum_{j=1}^m b_j(a_{j,k}+E_{j,k})
\label{eq:baseline_rescaled_output}.
\end{align}
In matrix form, we have $\hat{\mathbf c}^{\mathrm{base}}= \mathbf b(\mathbf A+\mathbf E_{\mathrm A})$.

According to the physical principle, the instantaneous power dissipated by the memristor on array $s$ is 
\begin{equation}
    P_{j,k}^{\mathrm{mem}(s)} = G_{j,k}^{(s)}x_j^2 , \quad s\in\{+,-\}.
\end{equation}
The signal-dependent feedback-path power of the TIA connected to array $s$ of column $k$ is
\begin{equation}
P_k^{\mathrm{TIA}(s)} = \frac{\left(Z_k^{(s)}\right)^2}{r_{\mathrm T}}, \quad s\in\{+,-\}. 
\end{equation}

\vspace{-10pt}

\vspace{-5pt}
\section{MSE and Power analysis for baseline}\label{MSE_power_baseline}
\vspace{-5pt}
This section provides the MSE and power consumption analysis for baseline VMM.
We use the aggregate mean-squared error (MSE) to evaluate the computation accuracy. 
Throughout this paper, the input vector $\mathbf{b}$ and the matrix $\mathbf{A}$ are deterministic. 
Hence, all expectations are taken only with respect to the memristor programming noise, and are denoted by $\mathbb{E}_{\mathbf{N}}[\cdot]$.
For the power analysis, we assume that the positive and negative crossbars have separate TIAs at each output column. 
The total power consumption consists of the memristors and the TIAs power consumption.

\vspace{-5pt}
\begin{theorem}[Aggregate baseline MSE]
The aggregate MSE of baseline VMM is 
\begin{equation}
\begin{aligned}
\mathcal{E}_{\mathrm{base}}
&\triangleq
\mathbb{E}_{\mathbf{N}}
\left[
\left\|
\widehat{\mathbf{c}}^{\mathrm{base}}
-
\mathbf{c}
\right\|_2^2
\,\middle|\,
\mathbf{b},\mathbf{A}
\right]
=
2\tau_v^2
\|\mathbf{b}\|_2^2
\operatorname{Tr}
\left(
\mathbf{D}_{\mathrm A}^{-2}
\right).
\end{aligned}
\label{eq:baseline_error}
\end{equation}
\end{theorem}

\begin{proof}
The computation-error vector is
$
\widehat{\mathbf{c}}^{\mathrm{base}}-\mathbf{c}
=
\mathbf{b}\mathbf{E}_{\mathbf{A}}.
$
The error at the $k$-th output is consequently
$
\Delta c_k^{\mathrm{base}}
=
\widehat c_k^{\mathrm{base}}-c_k
=
\sum_{j=1}^{m}b_jE_{j,k}.
$
Because the programming noises associated with different
memristors are mutually independent, the cross terms vanish, and
the MSE of the $k$-th output is
\begin{equation}
\mathbb{E}_{\mathbf{N}}
\left[
\left(\Delta c_k^{\mathrm{base}}\right)^2
\,\middle|\,
\mathbf{b},\mathbf{A}
\right]
= \mathbb{E}_{\mathbf{N}}
\left[
\left(
\sum_{j=1}^{m}b_jE_{j,k}
\right)^2
\right]
=
\frac{2\tau_v^2}{\kappa_{A,k}^2}
\|\mathbf{b}\|_2^2.
\label{eq:baseline_per_column_mse}
\end{equation}
The aggregate output MSE of the baseline VMM is defined as 
\begin{equation}
\begin{aligned}
\mathcal{E}_{\mathrm{base}}
&\triangleq
\mathbb{E}_{\mathbf{N}}
\left[
\left\|
\widehat{\mathbf{c}}^{\mathrm{base}}
-
\mathbf{c}
\right\|_2^2
\,\middle|\,
\mathbf{b},\mathbf{A}
\right]
=
\sum_{k=1}^{n}
\mathbb{E}_{\mathbf{N}}
\left[
\left(\Delta c_k^{\mathrm{base}}\right)^2
\,\middle|\,
\mathbf{b},\mathbf{A}
\right]
\end{aligned}
\label{eq:baseline_error_def}
\end{equation}
Substituting \eqref{eq:baseline_per_column_mse} into \eqref{eq:baseline_error_def} gives the expression in \eqref{eq:baseline_error}.
\end{proof}

\begin{proposition}[Memristor power of baseline]
\label{prop:baseline_memristor_power}
\vspace{-5pt}
The mean power dissipated by all memristors in the two crossbars is
\begin{equation}
\begin{aligned}
\overline{P}^{\mathrm{mem}}_{\mathrm{base}}
&\triangleq
\sum_{j=1}^{m}\sum_{k=1}^{n}
\mathbb{E}_{\mathbf{N}}
\left[
x_j^2
\left(
g_{j,k}^{(+)}+g_{j,k}^{(-)} \right)
\,\middle|\,
\mathbf{x},\mathbf{A}
\right]\\
&=
2ng_{\min}\|\mathbf{x}\|_2^2
+
(\mathbf{x}^{\odot 2})^{\mathsf T}
\mathbf{A}_{\text{abs}}\mathbf{D}_{\mathrm A}\mathbf{1}_n,
\end{aligned}
\label{eq:total_memristor_power}
\end{equation}
where $\odot$ denotes the Hadamard product,
defined by $(\mathbf{u}\odot\mathbf{v})_j=u_jv_j$
for $\mathbf{u},\mathbf{v}\in\mathbb{R}^m$.
We define
$\mathbf{x}^{\odot 2}=\mathbf{x}\odot\mathbf{x}
=[x_1^2,\ldots,x_m^2]^{\mathsf T}$.
The matrix $\mathbf{A}_{\text{abs}}$ denotes the element-wise
absolute value of $\mathbf{A}$, and
$\mathbf{1}_n\in\mathbb{R}^n$ is the all-ones column vector.
\end{proposition}

\begin{proof}
The instantaneous power dissipated by the memristor on array
$s\in\{+,-\}$ is
$ P_{j,k}^{\mathrm{mem}(s)} = G_{j,k}^{(s)}x_j^2 $.
Because $x_j$ is deterministic and the programming noise is
zero-mean,
\begin{equation}
\begin{aligned}
\mathbb{E}_{\mathbf{N}}
\left[
P_{j,k}^{\mathrm{mem}(s)}
\,\middle|\,
\mathbf{x},\mathbf{A}
\right]
&=
x_j^2
\mathbb{E}_{\mathbf{N}}
\left[
G_{j,k}^{(s)}
\,\middle|\,
\mathbf{A}
\right]
=
x_j^2g_{j,k}^{(s)}.
\end{aligned}
\label{eq:single_memristor_mean_power}
\end{equation}
The total mean power dissipated by the positive--negative memristor pair at position $(j,k)$ is
$P_{j,k}^{\mathrm{mem}}
=
P_{j,k}^{\mathrm{mem}(+)}
+
P_{j,k}^{\mathrm{mem}(-)}$.
Therefore,
\begin{equation}
\begin{aligned}
\mathbb{E}_{\mathbf{N}}
\left[
P_{j,k}^{\mathrm{mem}}
\,\middle|\,
\mathbf{x},\mathbf{A}
\right]
&=
x_j^2
\left(
g_{j,k}^{(+)}+g_{j,k}^{(-)}
\right)\\
&=
x_j^2
\left(
2g_{\min}
+
[a_{j,k}^{\mathrm{sc}}]_+
+
[a_{j,k}^{\mathrm{sc}}]_-
\right)\\
&=
x_j^2
\left(
2g_{\min}
+
\kappa_{A,k}|a_{j,k}|
\right),
\end{aligned}
\label{eq:memristor_pair_power}
\end{equation}
where $[a]_++[a]_- = |a|$ and $\kappa_{A,k}>0$ have been used.
Summing \eqref{eq:memristor_pair_power} over all rows and
columns gives
\begin{equation}
P^{\mathrm{mem}}\triangleq
\sum_{j=1}^{m}\sum_{k=1}^{n}
\mathbb{E}_{\mathbf{N}}
\left[
P_{j,k}^{\mathrm{mem}}
\,\middle|\,
\mathbf{x},\mathbf{A}
\right]
\end{equation}
which proves \eqref{eq:total_memristor_power}.
\end{proof}

\vspace{-10pt}
\begin{proposition}[TIA power of baseline]
\label{prop:baseline_tia_power}
The power consumption of all $2n$ TIAs can be
written directly in terms of $\mathbf{A}$ and $\mathbf{A}_{\text{abs}}$ as
\begin{equation}
\begin{aligned}
\overline{P}^{\mathrm{TIA}}_{\mathrm{base}}
&\triangleq \sum_{k=1}^n\mathbb{E}_{\mathbf{N}}
\left[
P_k^{\mathrm{TIA}(+)}
+
P_k^{\mathrm{TIA}(-)}
\middle|
\mathbf{x},\mathbf{A}
\right] 
\\
&{}=
2 r_T \left\|
g_{\min}S_x\mathbf{1}_n^{\mathsf T}
+
\frac{1}{2}
\mathbf{x}
\mathbf{A}_{\text{abs}}
\mathbf{D}_{\mathrm A}
\right\|_2^2 \\
&+
\frac{r_T \left\|
\mathbf{x}
\mathbf{A}
\mathbf{D}_{\mathrm A}
\right\|_2^2}{2}
+
2n r_T \tau_v^2
\|\mathbf{x}\|_2^2 ,
\end{aligned}
\label{eq:total_tia_power_matrix}
\end{equation}
where $ S_x \triangleq \sum_{j=1}^{m}x_j  = \mathbf{1}_m^{\mathsf T}\mathbf{x} $.
\end{proposition}

\begin{proof}
The signal-dependent feedback-path power of the TIA connected to array $s\in\{+,-\}$ of column $k$ is
$
P_k^{\mathrm{TIA}(s)} = \frac{\left(Z_k^{(s)}\right)^2}{r_T}. 
$
For compactness, define
$ S_x \triangleq \sum_{j=1}^{m}x_j  = \mathbf{1}_m^{\mathsf T}\mathbf{x}, $
$ Q_x \triangleq \sum_{j=1}^{m}x_j^2 = \|\mathbf{x}\|_2^2. $
The deterministic target current component and the random programming noise component are defined respectively as
\begin{equation}\label{mu_k_s_eta_k_s}
\mu_k^{(s)}
\triangleq
\sum_{j=1}^{m}x_jg_{j,k}^{(s)},
\qquad
\eta_k^{(s)}
\triangleq
\sum_{j=1}^{m}x_jN_{j,k}^{(s)}.
\end{equation}
Then substituting \eqref{eq:baseline_array_output}, \eqref{mu_k_s_eta_k_s}, and \eqref{eq:noisy_conductance_general} into original power expression, we can write
$
Z_k^{(s)}
=
r_T
\left(
\mu_k^{(s)}+\eta_k^{(s)}
\right).
$
Since $N_{j,k}^{(s)}$ are independent and identically distributed (i.i.d.) random variables with mean 0 and variance $\tau_v^2$, we have $\mathbb{E}_{\mathbf{N}} \left[\eta_k^{(s)}\right] = 0 $,
$
\operatorname{Var}_{\mathbf{N}}
\left(\eta_k^{(s)}\right)=
\sum_{j=1}^{m}
x_j^2
\operatorname{Var}_{\mathbf{N}}
\left(N_{j,k}^{(s)}\right)=
\tau_v^2
\sum_{j=1}^{m}x_j^2 = \tau_v^2Q_x.
$
Then the mean power of the $k$-th column of the TIA of array $s$ is 
\begin{equation}
\begin{aligned}
&\mathbb{E}_{\mathbf{N}}
\left[
P_k^{\mathrm{TIA}(s)}
\,\middle|\,
\mathbf{x},\mathbf{A}
\right]
= \frac{1}{r_T} \mathbb{E}_{\mathbf{N}} \left[ \left(Z_k^{(s)}\right)^2 \,\middle|\, \mathbf{x},\mathbf{A} \right] \\
&=
r_T \left[ \left(\mu_k^{(s)}\right)^2 + \tau_v^2Q_x
\right].
\end{aligned}
\label{eq:single_tia_mean_power}
\end{equation}
which follows from $\mathbb{E}[Y^2]=\operatorname{Var}(Y)+(\mathbb{E}[Y])^2$.
Substituting \eqref{eq:scaled_weight}, \eqref{eq:scaling_map} into \eqref{mu_k_s_eta_k_s}, and using $[a]_+=(|a|+a)/2$ and $[a]_-=(|a|-a)/2$, the deterministic target current components are 
\begin{equation}
\mu_k^{(\pm)}
= g_{\min}S_x +\frac{\kappa_{A,k}}{2} 
\left(
\sum_{j=1}^{m}x_j|a_{j,k}| \pm \sum_{j=1}^{m}x_ja_{j,k} \right).
\label{eq:target_currents_pm}
\end{equation}
Therefore, the mean power of TIA pair of the $k$-th column is 
\begin{equation}\label{eq:tia_pair_power_proof}
\begin{aligned}
&\mathbb{E}_{\mathbf{N}}
\left[
P_k^{\mathrm{TIA}}
\middle|
\mathbf{x},\mathbf{A}
\right]
= \mathbb{E}_{\mathbf{N}}
\left[
P_k^{\mathrm{TIA}(+)}
+
P_k^{\mathrm{TIA}(-)}
\middle|
\mathbf{x},\mathbf{A}
\right] \\
&=r_T \left[ \left(\mu_k^{(+)}\right)^2 +\left(\mu_k^{(-)}\right)^2
+ 2\tau_v^2Q_x \right]. 
\end{aligned}
\end{equation}
By substituting \eqref{eq:target_currents_pm} into \eqref{eq:tia_pair_power_proof}, then defining $u_k = g_{\min}S_x
+ \frac{\kappa_{A,k}}{2} \left( \sum_{j=1}^{m}x_j|a_{j,k}| \right)$, $v_k = \frac{\kappa_{A,k}}{2} \left( \sum_{j=1}^{m}x_ja_{j,k} \right)$, and using $(u_k+v_k)^2+(u_k-v_k)^2=2u_k^2+2v_k^2$,
the mean power of the TIA pair associated with column $k$ is
\begin{align}
&\mathbb{E}_{\mathbf{N}}
\left[
P_k^{\mathrm{TIA}}
 \middle|
\mathbf{x},\mathbf{A}
\right]
=
r_T\Bigg[
2\left(
g_{\min}S_x
+
\frac{\kappa_{A,k}}{2}
\sum_{j=1}^{m}x_j|a_{j,k}|
\right)^2 \notag \\ 
&+ \frac{\kappa_{A,k}^2}{2}
\left( \sum_{j=1}^{m}x_ja_{j,k} \right)^2 + 2\tau_v^2Q_x
\Bigg]. \label{eq:tia_pair_power}
\end{align}

Finally, the mean power of all TIAs is expressed as
\begin{equation}\label{P_TIA_baseline_original}
\begin{aligned}
P^{\mathrm{TIA}}
&= \sum_{k=1}^n \mathbb{E}_{\mathbf{N}}
\left[
P_k^{\mathrm{TIA}}
\,\middle|\,
\mathbf{x},\mathbf{A}
\right] =\sum_{k=1}^n  r_T
\left[ 2u_k^2 +  2v_k^2 + 2\tau_v^2Q_x
\right]
\end{aligned}
\end{equation}
Since the $k$-th entry of vector $\mathbf{x} \mathbf{A} \mathbf{D}_{A}$ is 
$
\kappa_{A,k}
\sum_{j=1}^{m}x_ja_{j,k},
$
and $k$-th entry of vector $\mathbf{x} |\mathbf{A}|\mathbf{D}_{A}$ is 
$
= \kappa_{A,k}
\sum_{j=1}^{m}x_j|a_{j,k}|,
$ 
by substituting \eqref{eq:tia_pair_power} into \eqref{P_TIA_baseline_original}, we obtain the matrix form expression in \eqref{eq:total_tia_power_matrix}.
\end{proof}

\vspace{-10pt}
\begin{theorem}[Total power consumption of baseline]
\label{thm:baseline_total_power}
The total power of the baseline VMM is 
\begin{equation}
\begin{aligned}
\overline{P}^{\mathrm{total}}_{\mathrm{base}}
&\triangleq
\overline{P}^{\mathrm{mem}}_{\mathrm{base}}
+
\overline{P}^{\mathrm{TIA}}_{\mathrm{base}},
\end{aligned}
\label{eq:baseline_total_power_closed_form}
\end{equation}
where $\overline{P}^{\mathrm{mem}}_{\mathrm{base}}$ and $\overline{P}^{\mathrm{TIA}}_{\mathrm{base}}$ are expressed as \eqref{eq:total_memristor_power} and
\eqref{eq:total_tia_power_matrix}, respectively.
\end{theorem}

\begin{proof}
Combining Propositions 1 and 2 gives the baseline total mean power.
\end{proof}

\vspace{-5pt}

\vspace{-10pt}
\section{MSE and Power analysis for extended FLAS}\label{MSE_power_FLAS}
\label{sec:flas_analysis}
\vspace{-7pt}
Our previous work \cite{lu2026low} proposes the factorized low-rank approximation scheme (FLAS) and evaluate the MSE gain without taking the scaling into account.
In this paper, we further extend the FLAS to support adjustable conductance scaling and derives the MSE and power consumption expressions of it.

\vspace{-8pt}
\subsection{Extended FLAS of VMM}
\vspace{-5pt}
Let $\mathbf A=\mathbf U\boldsymbol\Sigma\mathbf V^{\mathsf T}$
have rank $r$, where $\mathbf U\in\mathbb R^{m\times m}$ and
$\mathbf V\in\mathbb R^{n\times n}$ are orthogonal, and
$\boldsymbol\Sigma\in\mathbb R^{m\times n}$ has singular
values $\sigma_1\ge\cdots\ge\sigma_r\ge0$.
For an integer $1\le k\le r$, let $\mathbf U_k$ and $\mathbf V_k$
contain the first $k$ columns of $\mathbf U$ and $\mathbf V$,
respectively, and let
$\boldsymbol\Sigma_k=\operatorname{diag}(\sigma_1,\ldots,\sigma_k)$.
Define the rank-$k$ approximation $\mathbf A_k=\mathbf L\mathbf R$
with
\begin{equation}
\mathbf L=\mathbf U_k\boldsymbol\Sigma_k^{1/2}
    \in\mathbb R^{m\times k},
\qquad
\mathbf R=\boldsymbol\Sigma_k^{1/2}\mathbf V_k^{\mathsf T}
    \in\mathbb R^{k\times n}.
\end{equation}
The corresponding target output is
$\mathbf c_k=\mathbf b\mathbf A_k=\mathbf b\mathbf L\mathbf R$.

We now implement the memristor-based computation in two successive
VMMs, first with $\mathbf L$ and then with $\mathbf R$.
Their positive and negative crossbar pairs require $2mk$ and $2kn$
memristors per computation, respectively.
The lower memristor requirement enables multiple independently programmed replicas under the same hardware budget, whose averaged outputs suppress programming errors.
This mitigates the impact of random programming noise to the final MSE.
The two VMM stages are described below.

1) Stage 1:
Following Section~\ref{SystemModel}, we map $\mathbf{L}$ to an $m\times k$ crossbar
pair with input $\mathbf x_L=\mathbf b$ and scaling matrix
$
\mathbf D_{\mathrm L}=\operatorname{diag}(\kappa_{\mathrm L,1},\ldots,\kappa_{\mathrm  L,k}).
$
For $\ell \in \{1,...,k\}$, set $\kappa_{\mathrm L,\ell} \triangleq\frac{g_{u,\ell}-g_{\min}}{|l_{\ell}^{\max}|}$,  where
$|l_\ell^{\max}|\triangleq
\max_{1\le p\le m}|(\mathbf L)_{p,\ell}|$
and $g_{u,\ell}\in(g_{\min},g_{\max}]$ is the Stage~1 upper
target conductance.

For a positive integer $t_{\mathrm{L}}$,
replica $i\in\{1,\ldots,t_{\mathrm{L}}\}$ computes
$\widehat{\mathbf c}_{\mathrm{L}}^{(i)}
=\mathbf x_{\mathrm{L}}(\mathbf L+\mathbf E_{\mathrm{L}}^{(i)})$,
where
$\mathbf E_{\mathrm{L}}^{(i)}
\triangleq
\mathbf N_{\mathrm{L}}^{\mathrm{eff},(i)}
\mathbf D_{\mathrm{L}}^{-1}$
and
$\mathbf N_{\mathrm{L}}^{\mathrm{eff},(i)}
\in\mathbb R^{m\times k}$
is the effective noise matrix defined by \eqref{eq:effective_noise} for that replica.
All underlying programming noises are mutually independent across
memristors, crossbars, replicas, and stages.
Thus, $\mathbf E_{\mathrm{L}}^{(i)}$ has zero-mean entries with
$\operatorname{Var}((\mathbf E_{\mathrm{L}}^{(i)})_{p,\ell})
=2\tau_v^2/\kappa_{\mathrm{L},\ell}^2$, $p\in\{1,\ldots,m\}$.
Averaging the replica outputs gives the final output of Stage 1:
\begin{equation}
\widehat{\mathbf c}_{\mathrm{L}}
=\frac{1}{t_{\mathrm{L}}}
 \sum_{i=1}^{t_{\mathrm{L}}}
 \widehat{\mathbf c}_{\mathrm{L}}^{(i)}
=\mathbf x_{\mathrm{L}}
 (\mathbf L+\overline{\mathbf E}_{\mathrm{L}}),
\end{equation}
where
$\overline{\mathbf E}_{\mathrm{L}}
\triangleq
t_{\mathrm{L}}^{-1}
\sum_{i=1}^{t_{\mathrm{L}}}\mathbf E_{\mathrm{L}}^{(i)}$.

2) Stage 2:
Map $\mathbf R$ to a $k\times n$ crossbar pair with input
$\mathbf x_{\mathrm{R}}
=\widehat{\mathbf c}_{\mathrm{L}}\in\mathbb R^{1\times k}$
and scaling matrix
$\mathbf D_{\mathrm{R}}
=\operatorname{diag}(\kappa_{\mathrm{R},1},\ldots,\kappa_{\mathrm{R},n})$.
For $q \in \{1,...,n\}$, set
$\kappa_{\mathrm{R},q}
=(g_{u,q}-g_{\min})/|r_q^{\max}|$,
where
$|r_q^{\max}|\triangleq
\max_{1\le\ell\le k}|(\mathbf R)_{\ell,q}|$
is assumed nonzero.
The Stage~2 upper targets $g_{u,q}\in(g_{\min},g_{\max}]$
are chosen separately.

For a positive integer $t_{\mathrm{R}}$,
replica $j\in\{1,\ldots,t_{\mathrm{R}}\}$ produces
$\widehat{\mathbf c}_{\mathrm{R}}^{(j)}
=\mathbf x_{\mathrm{R}}(\mathbf R+\mathbf E_{\mathrm{R}}^{(j)})$,
where
$\mathbf E_{\mathrm{R}}^{(j)}
\triangleq
\mathbf N_{\mathrm{R}}^{\mathrm{eff},(j)}
\mathbf D_{\mathrm{R}}^{-1}$
and
$\mathbf N_{\mathrm{R}}^{\mathrm{eff},(j)}
\in\mathbb R^{k\times n}$
is defined by \eqref{eq:effective_noise} for that replica.
The entries of $\mathbf E_{\mathrm{R}}^{(j)}$ also have zero mean
and variance
$\operatorname{Var}((\mathbf E_{\mathrm{R}}^{(j)})_{\ell,q})
=2\tau_v^2/\kappa_{\mathrm{R},q}^2$, $\ell=1,\ldots,k$.
Averaging $\widehat{\mathbf c}_{\mathrm{R}}^{(j)}$ yields the final of output of FLAS
\begin{equation}
\widehat{\mathbf c}^{\mathrm{FLAS}}
=\frac{1}{t_{\mathrm{R}}}
 \sum_{j=1}^{t_{\mathrm{R}}}
 \widehat{\mathbf c}_{\mathrm{R}}^{(j)}
=\mathbf x_{\mathrm{L}}
 (\mathbf L+\overline{\mathbf E}_{\mathrm{L}})
 (\mathbf R+\overline{\mathbf E}_{\mathrm{R}}),
\end{equation}
where
$\overline{\mathbf E}_{\mathrm{R}}
\triangleq
t_{\mathrm{R}}^{-1}
\sum_{j=1}^{t_{\mathrm{R}}}\mathbf E_{\mathrm{R}}^{(j)}$.

Assume that the budget of memristors is $2mn$, then the replica count $t_{\mathrm L}$ and $t_{\mathrm R}$ should satisfy the following constraint
\begin{equation}
2t_{\mathrm L}mk+2t_{\mathrm R}kn \le 2mn.
\label{eq:flas_device_budget}
\end{equation}

\subsection{Second-stage input statistics}

Let $\boldsymbol{\mu}=\mathbf{x}\mathbf{L}$ and
$\boldsymbol{\xi}=\mathbf{x}\overline{\mathbf{E}}_L$, both
in $\mathbb{R}^{1\times k}$.
Then $\mathbf{x}_R=\boldsymbol{\mu}+\boldsymbol{\xi}$,
$\mathbb{E}_N[\boldsymbol{\xi}]=\mathbf{0}$.
Define the noise covariance and input second-moment matrices,
respectively, as
\begin{equation}
\begin{aligned}
\mathbf{C}_L
&\triangleq\mathbb{E}_N[\boldsymbol{\xi}^{T}\boldsymbol{\xi}]
=\frac{2\tau_v^2Q_x}{t_L}\mathbf{D}_L^{-2},\\
\boldsymbol{\Psi}_R
&\triangleq\mathbb{E}_N[\mathbf{x}_R^{T}\mathbf{x}_R]
=\boldsymbol{\mu}^{T}\boldsymbol{\mu}+\mathbf{C}_L.
\end{aligned}
\label{eq:flas_input_moments}
\end{equation}
where $\mathbf{C}_L, \boldsymbol{\Psi}_R \in \mathbb{R}^{k\times k}$.
For $\ell=1,\ldots,k$, let
\begin{equation}
\begin{aligned}
(\boldsymbol{\Psi}_R)_{\ell,\ell} \triangleq \nu_\ell
=\mathbb{E}_N[x_{R,\ell}^{2}]
 =\mu_\ell^2+\frac{2\tau_v^2Q_x}{t_L\kappa_{L,\ell}^2}.
\end{aligned}
\label{eq:flas_input_element_moments}
\end{equation}
Collect these moments in the row vector
$\boldsymbol{\nu}=(\nu_1,\ldots,\nu_k)$.
Define the expected squared norm of the second-stage input by
\begin{equation}
\begin{aligned}
Q_R&\triangleq\mathbb{E}_N[\|\mathbf{x}_R\|_2^2]=\sum_{\ell=1}^{k}\nu_\ell
=\operatorname{Tr}(\boldsymbol{\Psi}_R)\\
&=\|\mathbf{x}\mathbf{L}\|_2^2
+\frac{2\tau_v^2Q_x}{t_L}\operatorname{Tr}(\mathbf{D}_L^{-2}).
\end{aligned}
\label{eq:flas_second_stage_Q}
\end{equation}
Thus $Q_R$ is deterministic and includes the contribution
from first-stage programming noise.
These statistics require only the specified independence and first two noise moments.

\subsection{MSE of FLAS}

\begin{theorem}[Aggregate FLAS MSE]
\label{thm:flas_mse}
The aggregate MSE of FLAS is
\begin{align}\label{eq:flas_mse}
&\mathcal{E}_{\mathrm{FLAS}}
\triangleq \mathbb{E}_N\!\left[
\|\widehat{\mathbf{c}}^{\mathrm{FLAS}}-\mathbf{c}\|_2^2
\right]  \\
&=\|\mathbf{x}(\mathbf{A}-\mathbf{A}_k)\|_2^2
+\frac{2\tau_v^2Q_x}{t_L}
\operatorname{Tr}(\mathbf{D}_L^{-2}\boldsymbol{\Sigma}_k)
+\frac{2\tau_v^2Q_R}{t_R}
\operatorname{Tr}(\mathbf{D}_R^{-2}), \notag
\end{align}
where $ Q_{\mathrm L} \triangleq \sum_{j=1}^{m}x_{\mathrm L, j}^2 = \|\mathbf{x}_{\mathrm L}\|_2^2 $,
and $Q_{\mathrm R}$ is the expected squared norm of the second-stage input, i.e.,
\begin{equation}
\begin{aligned}
Q_{\mathrm R} \triangleq\mathbb{E}_N[\|\mathbf{x}_{\mathrm R}\|_2^2]
=\|\mathbf{x}_{\mathrm L}\mathbf{L}\|_2^2
+\frac{2\tau_v^2Q_{\mathrm L}}{t_{\mathrm L}}\operatorname{Tr}(\mathbf{D}_{\mathrm L}^{-2}).
\end{aligned}
\label{eq:flas_second_stage_Q}
\end{equation}
\end{theorem}

\begin{proof}
The output error decomposes as
\begin{equation}
\begin{aligned}
\widehat{\mathbf{c}}^{\mathrm{FLAS}}-\mathbf{c}
&=\mathbf{x}(\mathbf{A}_k-\mathbf{A})+\boldsymbol{\xi}\mathbf{R} +\mathbf{x}_R\overline{\mathbf{E}}_R.
\end{aligned}
\label{eq:flas_error_decomposition}
\end{equation}
Because $\mathbb{E}_N[\boldsymbol{\xi}]=\mathbf{0}$ and
$\overline{\mathbf{E}}_R$ is zero-mean and independent of
the first-stage noises, all cross terms vanish in expectation.
The second-stage noise statistics give
\begin{equation}
\mathbb{E}_N\!\left[
\overline{\mathbf{E}}_R\overline{\mathbf{E}}_R^{T}
\right]
=\frac{2\tau_v^2}{t_R}
\operatorname{Tr}(\mathbf{D}_R^{-2})\mathbf{I}_k.
\label{eq:flas_stage2_noise_moment}
\end{equation}
Using independence,
$\mathbb{E}_N[\boldsymbol{\xi}^{T}\boldsymbol{\xi}]
=\mathbf{C}_L$, and
$\mathbb{E}_N[\|\mathbf{x}_R\|_2^2]=Q_R$, we obtain
\begin{equation}
\begin{aligned}
\mathcal{E}_{\mathrm{FLAS}}(\mathbf{x})
&=\|\mathbf{x}(\mathbf{A}-\mathbf{A}_k)\|_2^2
+\operatorname{Tr}
(\mathbf{C}_L\mathbf{R}\mathbf{R}^{T})
+\frac{2\tau_v^2Q_R}{t_R}
\operatorname{Tr}(\mathbf{D}_R^{-2}).
\end{aligned}
\label{eq:flas_mse_moments}
\end{equation}
According to the SVD factorization, we have
\begin{equation}
\mathbf{R}\mathbf{R}^{T}
=\boldsymbol{\Sigma}_k^{1/2}
\mathbf{V}_k^{T}\mathbf{V}_k
\boldsymbol{\Sigma}_k^{1/2}
=\boldsymbol{\Sigma}_k.
\label{eq:flas_R_gram}
\end{equation}
Substituting this identity and
$\mathbf{C}_L=(2\tau_v^2Q_x/t_L)\mathbf{D}_L^{-2}$
into~\eqref{eq:flas_mse_moments} proves~\eqref{eq:flas_mse}.
The interaction between the two stages' programming noises
is included through $Q_R$.
\end{proof}

\vspace{-5pt}
\subsection{Power consumption analysis of FLAS}
\vspace{-5pt}
As in Section~\ref{MSE_power_baseline}, the total power consumption  memristor power consumption and TIA power consumption. 

\begin{proposition}[First-stage memristor power]
\label{prop:flas_L_mem}
The mean power of all first-stage memristors is
\begin{equation}
\overline{P}_{\mathrm L}^{\mathrm{mem}}
=t_{\mathrm L}\left[2kg_{\min}Q_{\mathrm L}
+\mathbf{x}^{\odot2}\mathbf{L}_{\text{abs}}\mathbf{D}_{\mathrm L}\mathbf{1}_k
\right].
\label{eq:flas_L_mem}
\end{equation}
where $\mathbf{L}_{\text{abs}}$ denotes the element-wise absolute value of $\mathbf{L}$, and $\mathbf{1}_k\in\mathbb{R}^k$ is the all-ones column vector.
\end{proposition}

\begin{proof}
Apply Proposition~1 to one first-stage replica and sum over
the $t_L$ replicas.
\end{proof}

\begin{proposition}[First-stage TIA power]
\label{prop:flas_L_TIA}
The mean power of all $2kt_{\mathrm L}$ first-stage TIAs is
\begin{equation}
\begin{aligned}
\overline{P}_{\mathrm L}^{\mathrm{TIA}}
&=2t_{\mathrm L}r_T\|\mathbf{x}_{\mathrm L}\mathbf{H}_{\mathrm L}\|_2^2
+\frac{t_{\mathrm L}r_T}{2}
\|\mathbf{x}_{\mathrm L}\mathbf{L}\mathbf{D}_{\mathrm L}\|_2^2
+2t_{\mathrm L}kr_T\tau_v^2Q_{\mathrm L}.
\end{aligned}
\label{eq:flas_L_TIA}
\end{equation}
where $ \mathbf{H}_{\mathrm L}=g_{\min}\mathbf{I}_{\mathrm L}
+\frac{1}{2}\mathbf{L}_{\text{abs}}\mathbf{D}_{\mathrm L}$.
The $\mathbf{I}_{\mathrm R} \in \mathbb R^{m\times k}$ denotes an all-ones matrix.
\end{proposition}

\begin{proof}
Apply Proposition~2 to $\mathbf{L}$ and $\mathbf{D}_L$, and
multiply by $t_L$. Each physical replica dissipates power
before output averaging.
\end{proof}

\begin{proposition}[Second-stage memristor power]
\label{prop:flas_R_mem}
The mean power of all second-stage memristors is
\begin{equation}
\overline{P}_{\mathrm R}^{\mathrm{mem}}
=t_{\mathrm R}\left[2ng_{\min}Q_{\mathrm R}
+\boldsymbol{\nu}\mathbf{R}_{\text{abs}}\mathbf{D}_{\mathrm R}\mathbf{1}_n
\right].
\label{eq:flas_R_mem}
\end{equation}
where $\mathbf{R}_{\text{abs}}$ denotes the element-wise absolute value of $\mathbf{R}$.
The $\ell$-th entry of the row vector $\boldsymbol{\nu}$ is
$\nu_\ell \triangleq \mathbb{E}_N[x_{\mathrm{R},\ell}^{2}]$.
\end{proposition}

\begin{proof}
The second-stage programming noises are independent of
$\mathbf{x}_R$. Condition on $\mathbf{x}_R$, apply
Proposition~1, and use
$\mathbb{E}_N[x_{R,\ell}^{2}]=\nu_\ell$ and
$\mathbb{E}_N[\|\mathbf{x}_R\|_2^2]=Q_R$.
\end{proof}

\begin{proposition}[Second-stage TIA power]
\label{prop:flas_R_TIA}
The mean power of all $2nt_{\mathrm R}$ second-stage TIAs is
\begin{equation}
\begin{aligned}
\overline{P}_{\mathrm R}^{\mathrm{TIA}}
&=2t_{\mathrm R}r_T\operatorname{Tr}
(\mathbf{H}_{\mathrm R}^{\mathsf T}\boldsymbol{\Psi}_{\mathrm R}\mathbf{H}_{\mathrm R})\\
&\quad+\frac{t_{\mathrm R}r_T}{2}\operatorname{Tr}
(\mathbf{D}_{\mathrm R}\mathbf{R}^{\mathsf T}\boldsymbol{\Psi}_{\mathrm R}
\mathbf{R}\mathbf{D}_{\mathrm R})
+2t_{\mathrm R}nr_T\tau_v^2Q_{\mathrm R}.
\end{aligned}
\label{eq:flas_R_TIA}
\end{equation}
where
$\mathbf H_{\mathrm{R}}
=g_{\min}\mathbf I_{\mathrm{R}}
+\frac{1}{2}\mathbf R_{\mathrm{abs}}\mathbf D_{\mathrm{R}}$,
$\mathbf I_{\mathrm{R}}\in\mathbb R^{k\times n}$
is the all-ones matrix, and
$\boldsymbol{\Psi}_{\mathrm{R}}
\triangleq\mathbb{E}_N[
\mathbf x_{\mathrm{R}}^{\mathsf T}\mathbf x_{\mathrm{R}}]$
is the second-moment matrix of $\mathbf x_{\mathrm{R}}$.
\end{proposition}

\begin{proof}
Conditioning on $\mathbf{x}_R$ reduces each replica to
Proposition~2. Take expectations using
$\mathbb{E}_N[\|\mathbf{x}_R\mathbf{B}\|_2^2]
=\operatorname{Tr}(\mathbf{B}^{T}\boldsymbol{\Psi}_R\mathbf{B})$,
where $\mathbf{B} = 1/2 \mathbf{R} \mathbf{D}_R$,
and sum over the $t_R$ replicas.
\end{proof}

\begin{theorem}[Aggregate power consumption of FLAS]
\label{thm:flas_power}
The total power consumption of FLAS is
\begin{equation}
\begin{aligned}
\overline{P}_{\mathrm{FLAS}}^{\mathrm{total}}
&=\overline{P}_{\mathrm L}^{\mathrm{mem}}
+\overline{P}_{\mathrm L}^{\mathrm{TIA}}
+\overline{P}_{\mathrm R}^{\mathrm{mem}}
+\overline{P}_{\mathrm R}^{\mathrm{TIA}},
\end{aligned}
\label{eq:flas_total_power}
\end{equation}
where the four terms are given by Propositions~\ref{prop:flas_L_mem}--\ref{prop:flas_R_TIA}.
\end{theorem}

\begin{proof}
Sum the four component mean powers given by Propositions~\ref{prop:flas_L_mem}--\ref{prop:flas_R_TIA}.
\end{proof}

\vspace{-5pt}
\section{MSE and Power Consumption tradeoff}
\label{tradeoff}
\vspace{-2pt}
\subsection{Optimization Problem}
\vspace{-2pt}
We optimize power consumption with MSE constraint for the baseline and FLAS in this section.
We respectively denote vector forms of the scaling coefficients for $\mathbf A, \mathbf L, \mathbf R$, as $\boldsymbol{\kappa}_{\mathrm A}=(\kappa_{\mathrm A,1},...\kappa_{\mathrm A,n})$, $\boldsymbol{\kappa}_{\mathrm L}=(\kappa_{\mathrm L,1},...\kappa_{\mathrm L,k})$, and $\boldsymbol{\kappa}_{\mathrm R}=(\kappa_{\mathrm R,1},...\kappa_{\mathrm R,n})$.
For fixed $\mathbf A$ and vector $\mathbf b$,
we minimize the mean total power subject to an aggregate
MSE tolerance $\epsilon>0$.
Varying $\epsilon$ characterizes the power--MSE trade-off, since tightening the accuracy constraint shrinks the feasible set and cannot decrease the optimal power.

The baseline optimization problem is formulated as
\vspace{-5pt}
\begin{equation}
\begin{aligned}
\text{(P1)}\quad
&\min_{\boldsymbol{\kappa}_{\mathrm A}}\quad \overline{P}^{\mathrm{total}}_{\mathrm{base}}\\
&\text{s.t.}\quad \mathcal{E}_{\mathrm{base}}\le\epsilon,\\
&\phantom{\text{s.t.}\quad}
0<\kappa_{\mathrm A,q}\le\overline{\kappa}_{\mathrm A,q},
\quad q=1,\ldots,n.
\end{aligned}
\label{prob:P1}
\end{equation}
Here $\overline{P}^{\mathrm{total}}_{\mathrm{base}}$ and $\mathcal{E}_{\mathrm{base}}$ are given by Theorems~2 and~1, respectively. 
The conductance constraint in Section~2 gives the scaling upper bound
$\overline{\kappa}_{\mathrm{A},q}
\triangleq(g_{\max}-g_{\min})/|a_q^{\max}|$.

(P1) is convex because its objective is convex quadratic and
its aggregated MSE is a nonnegative weighted sum of
$\kappa_{\mathrm{A},q}^{-2}$.
It can therefore be solved globally whenever feasible.

Let $\mathcal{K}$ be the set of ranks $\{1,\ldots,r\}$. 
Jointly select the rank,
replication counts, and scaling factors, we can formulate the FLAS optimization problem as
\vspace{-5pt}
\begin{equation}
\begin{aligned}
\text{(P2)}\quad
&\min_{\substack{k,t_{\mathrm L},t_{\mathrm R},\\
\mathbf{\kappa}_{\mathrm L},\mathbf{\kappa}_{\mathrm R}}}
\quad \overline{P}_{\mathrm{FLAS}}^{\mathrm{total}}\\
&\text{s.t.}\quad \mathcal{E}_{\mathrm{FLAS}}\le\epsilon,\\
&\phantom{\text{s.t.}\quad} k(mt_L+nt_R)\le mn,\\
&\phantom{\text{s.t.}\quad}
k\in\mathcal K,\qquad
t_{\mathrm{L}},t_{\mathrm{R}}\in\{1,2,\ldots\},\\
&\phantom{\text{s.t.}\quad}
0<\kappa_{L,\ell}\le\overline{\kappa}_{L,\ell},
\quad \ell=1,\ldots,k,\\
&\phantom{\text{s.t.}\quad}
0<\kappa_{R,q}\le\overline{\kappa}_{R,q},
\quad q=1,\ldots,n.
\end{aligned}
\label{prob:P2}
\end{equation}
The $\overline{P}_{\mathrm{FLAS}}^{\mathrm{total}}$ and $\mathcal{E}_{\mathrm{FLAS}}$ are given by Theorems~\ref{thm:flas_power} and~\ref{thm:flas_mse}, respectively.
The scaling upper bounds follow from Section~4.1, 
$\overline{\kappa}_{\mathrm{L},\ell}
\triangleq(g_{\max}-g_{\min})/|l_\ell^{\max}|$
and
$\overline{\kappa}_{\mathrm{R},q}
\triangleq(g_{\max}-g_{\min})/|r_q^{\max}|$.

(P2) is a mixed-integer nonlinear program.
We enumerate all integer triples
$(k,t_{\mathrm{L}},t_{\mathrm{R}})$ satisfying the rank,
replication count, and memristor budget constraints.
The factors and their scaling bounds are recomputed for each
candidate rank.
For each triple, we numerically optimize
$\boldsymbol{\kappa}_{\mathrm{L}}$ and
$\boldsymbol{\kappa}_{\mathrm{R}}$
subject to the MSE and scaling constraints.
Among the feasible solutions obtained, we select the one
with the lowest total power.

\vspace{-5pt}
\subsection{Numerical results}
\vspace{-5pt}
We compare the power--MSE trade-offs of FLAS and baseline
using two matrix examples.
The parameters are
$g_{\min}=1\,\mu\mathrm{S}$,
$g_{\max}=100\,\mu\mathrm{S}$, and
$\tau_v^2=5\times10^{-7}$.
We construct two fixed full-rank matrices
$\mathbf A=\mathbf U\boldsymbol{\Sigma}\mathbf V^{\mathsf T}
\in\mathbb R^{24\times24}$
with the same randomly generated orthogonal factors
$\mathbf U$ and $\mathbf V$.
Their singular values are
$\sigma_i=3e^{-0.55(i-1)}$ and $\sigma_i=3/i$,
$i=1,\ldots,24$, corresponding to fast and moderate spectral
decay, respectively.
We draw one input vector with independent entries
$b_i\sim\mathcal U(0.02,0.20)$ and use the same realization
for both matrices.
Accuracy is measured by
$\mathrm{NMSE}
=\mathbb E_N[\|\widehat{\mathbf c}-\mathbf b\mathbf A\|_2^2]
 /\|\mathbf b\mathbf A\|_2^2$,
where the expectation is over programming noise.
We verified that the component magnitudes of $\mathbf b\mathbf U$
are reasonably balanced, with no exceptionally large or small values.
The comparison is therefore not dominated by atypical input
alignment with particular left singular vectors.

Fig.~\ref{baseline_FLAS_nmse_compare} compares the power--MSE
trade-offs obtained under the common memristor budget $2mn$.
We set $r_T=5\,\mathrm{k}\Omega$ for this figure.
FLAS reduces power by $61.7\%$ and $32.5\%$ at the marked
operating points for fast and moderate spectral decay, respectively.
For the fast-decaying matrix, FLAS can satisfy the MSE tolerance
with aggressive low-rank truncation, retaining only a small
number $k$ of singular components.
This substantially reduces the memristor count required for
the two-stage implementation and enables lower power.
Consequently, FLAS outperforms baseline throughout the common
feasible NMSE range shown.
For the moderately decaying matrix, the curves cross near
NMSE $0.074$, below which baseline consumes less power.
At smaller NMSE, FLAS needs the larger rank $k$ or higher conductance
levels to maintain accuracy, which can offset the power savings
from factorization.
At the marked operating points, $t_L = 1, t_R = 1$ implies that the reported savings are achieved without replication.

\begin{figure}[htbp]
	\centering
    \vspace{-5pt}
	\includegraphics[width = 0.9 \linewidth]{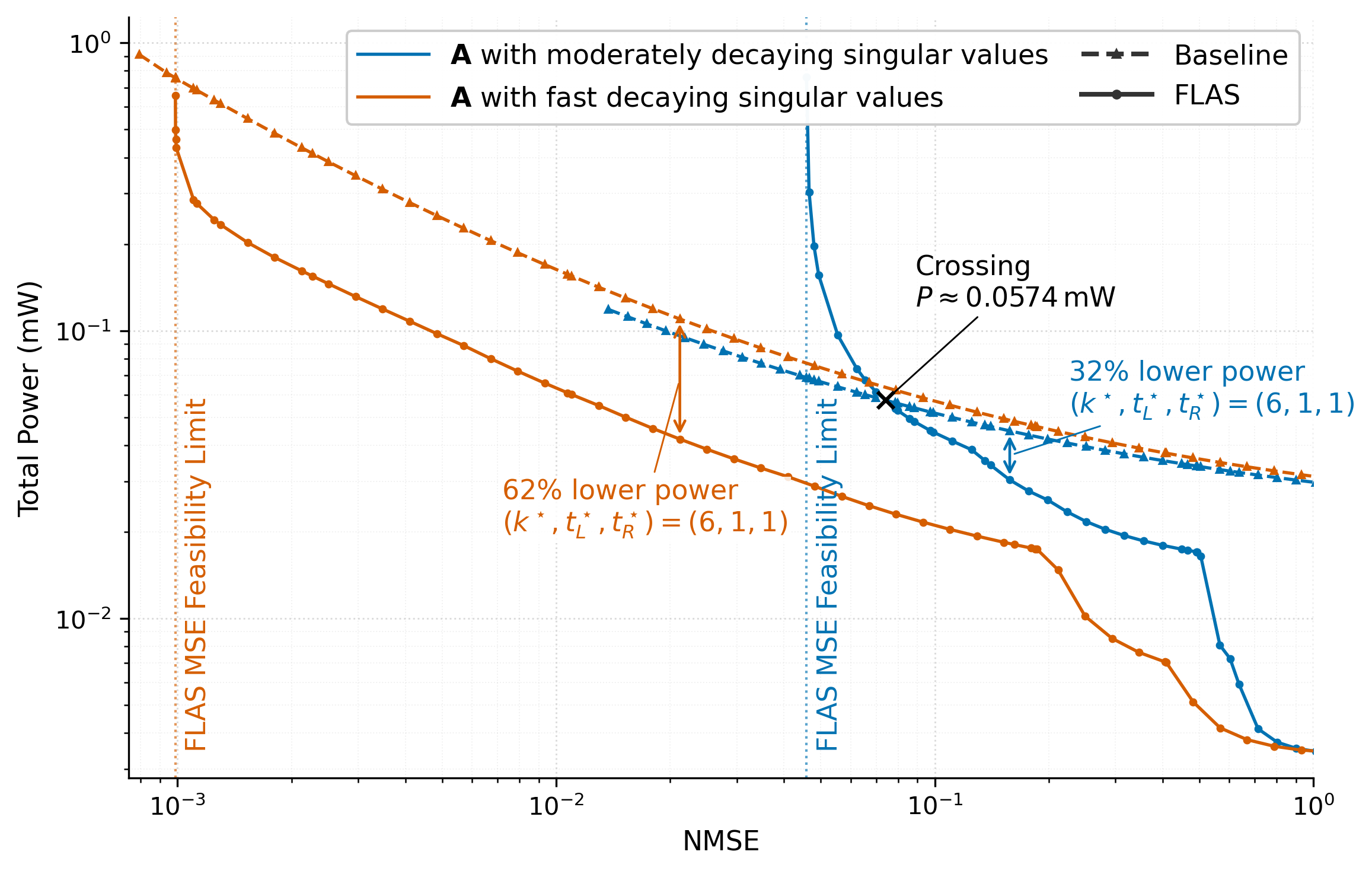}
	\vspace{-10pt}
	\caption{Power--NMSE trade-offs of FLAS and baseline
under a common memristor budget $2mn$.\hspace{50pt}}
    \vspace{-5pt}
	\label{baseline_FLAS_nmse_compare}
\end{figure}

Fig.~\ref{differ_r_T} examines how the TIA feedback resistance
affects total power and the fraction of power consumed by TIAs.
For the moderately decaying matrix, we select two pairs of
operating points from the FLAS and baseline curves in
Fig.~\ref{baseline_FLAS_nmse_compare}, at NMSE limits
$0.06$ and $0.10$.
These limits lie near and on opposite sides of the crossing
at approximately $0.074$.
The two pairs therefore allow us to examine how feedback
resistance affects these opposite power advantages.
Keeping the matrix, input, and NMSE limits fixed, we sweep
$r_T$ over ten logarithmically spaced values from
$1$ to $50\,\mathrm{k}\Omega$ and optimize power of both schemes
at each $r_T$. 
The reliable range of $r_T$ in practive is usually $5$ to $20\,\mathrm{k}\Omega$.
At NMSE $0.10$, FLAS has a smaller TIA power fraction.
At NMSE $0.06$, baseline has the smaller TIA power fraction.
For both FLAS and baseline, the relative total power advantage increases with $r_T$.
These results show that the NMSE tolerance determines which scheme
consumes less power, while increasing $r_T$ strengthens the advantage
of the scheme with the smaller TIA power fraction.

\begin{figure}[htb]
\begin{minipage}[b]{0.48\linewidth}
  \centering
  \centerline{\includegraphics[width=5.0cm]{  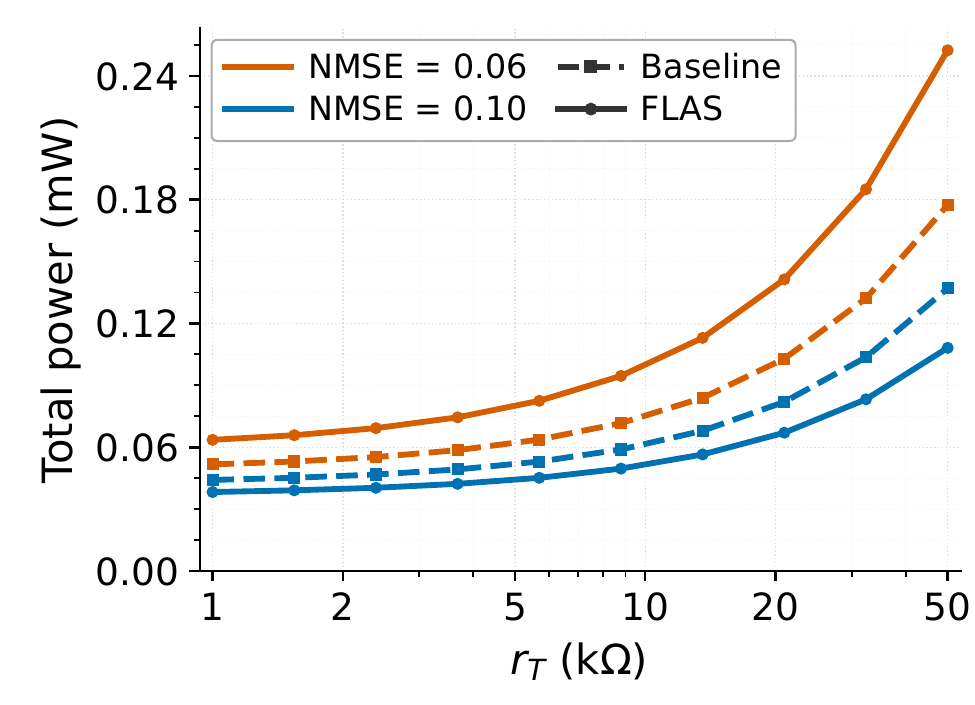}}
  \centerline{(a)}\medskip
\end{minipage}
\hfill
\begin{minipage}[b]{0.49\linewidth}
  \centering
  \centerline{\includegraphics[width=5.0cm]{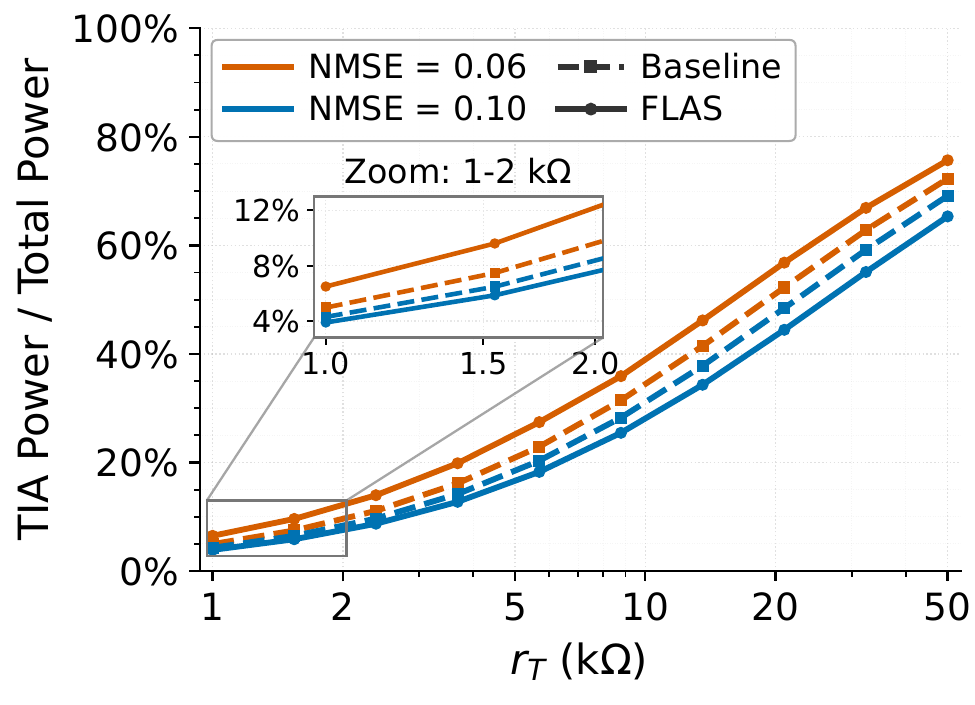}}
\centerline{(b)}\medskip
\end{minipage}

 \vspace{-10pt}
\caption{Effect of TIA feedback resistance for the moderately
decaying matrix at NMSE limits $0.06$ and $0.10$.
(a) Total power. (b) Fraction of total power consumed by TIAs.}
\label{differ_r_T}
\end{figure}

\vspace{-15pt}
\section{Conclusion}

We extended FLAS with adjustable conductance scaling and provides the power--MSE trade-off analysis of FLAS and baseline.
The numerical results show that FLAS's power advantage over baseline depends on matrix spectral decay and MSE.
The TIA power fractions explain how this advantage changes with feedback resistance.

\newpage
\bibliographystyle{IEEEbib}
\bibliography{refs_new}
\end{document}